\documentclass{amsart}
\usepackage{amssymb,amsmath,mathtools,amsthm}
\usepackage{graphicx} % Required for inserting images
\newtheorem{prop}{Proposition}
\theoremstyle{remark}
\newtheorem{remark}{Remark}
\numberwithin{equation}{section}

\begin{document}
\title[Real eigenvalue asymptotics]{Asymptotics of the real eigenvalue distribution \\for the real spherical ensemble}
\author{Peter J. Forrester}
\address{School of Mathematical and Statistics, ARC Centre of Excellence for Mathematical and Statistical Frontiers, The University of Melbourne, Victoria 3010, Australia}
\email{pjforr@unimelb.edu.au}
\date{}

\begin{abstract}
    The real Ginibre spherical ensemble consists of random matrices of the form $A B^{-1}$, where $A,B$ are independent standard real Gaussian $N \times N$ matrices. 
    The expected number of real eigenvalues is known to be of order $\sqrt{N}$. We consider the probability $p_{N.M}^{\rm r}$ that there are $M$ real eigenvalues in various regimes. These are when $M$ is proportional to $N$ (large deviations), when $N$ is proportional to $\sqrt{N}$ (intermediate deviations), and when $M$ is in the neighbourhood of the mean (local central limit theorem). This is done using a Coulomb gas formalism in the large deviations case, and by determining the leading asymptotic form of the generating function for the probabilities in the case of intermediate deviations (the local central limit regime was known from earlier work). Moreover a matching of the left tail asymptotics of the intermediate deviation regime with that of the right tail of the large deviation regime
    is exhibited,
    as is a matching of the right tail intermediate deviation regime with the leading order form of the probabilities in the local central limit regime.
    We also give the leading asymptotic form of $p_{N,0}^{\rm r}$, i.e.~the probability of no real eigenvalues.
\end{abstract}

\maketitle

\section{Introduction}
Consider a square random matrix with real, zero mean, unit standard deviation i.i.d~elements. Experiments using suitable computer software show that the expected number of real eigenvalues is of order $\sqrt{N}$, and that the probability $p_{N,M}^{\rm r}$ that there are $M$ real eigenvalues (since complex eigenvalues occur in complex conjugate pairs, we require $M$ to have the same parity as $N$) is nonzero. Such simple to observe effects were first quantified in the case of standard Gaussian
(also known as GinOE \cite{BF25}) random matrices in \cite{EKS94,Ed97}. Thus from \cite{EKS94} we have that asymptotically the expected number of real eigenvalues is $\sqrt{2N \over \pi}$, while \cite{Ed97} established the exact result $p_{N,N}^{\rm r} = 2^{-N(N-1)/4}$ for the probability that all eigenvalues are real. Later this first result was extended beyond the Gaussian case, being shown to remain valid at least for the circumstance that the first four moments of the distribution of the entries match a standard normal \cite{TV15}; see the informative introduction in \cite{Lu18} for more on this, and also for a listing of remaining outstanding questions relating to universality. In relation to the second result, the subsequent work
\cite{AK07} gave the asymptotic expansion of $p_{N,N-2}^{\rm r}/p_{N,N}^{\rm r}$ up to terms which go to zero in the limit.

Outside of the class of random matrices with real i.i.d.~elements there are now known to be several ``integrable'' (or exactly solvable) ensembles which permit a detailed analysis of the number or real eigenvalues. Specifically, in historical order and restricting attention  to the expected number, or exact or asymptotic results for individual $p_{N,M}^{\rm r}$, this has shown to be possible for the real spherical ensemble of random matrices of the form $A B^{-1}$, with $A,B$ independent GinOE matrices \cite{EKS94,FM11}; for the elliptic GinOE ensemble of random matrices $\sqrt{1+\tau} S + \sqrt{1 - \tau} A$, with $S$ ($A$) a symmetric (antisymmetric) standard Gaussian matrix
and $\tau \in [0,1)$ a parameter \cite{FN08p,Ta22,BKLL23,BL24,ABL25}; for truncations of Haar distributed real orthogonal matrices \cite{KSZ09,PS18}; for products of GinOE matrices
(Haar distributed real orthogonal matrices)
\cite{Fo14,Ku15,FI16,Si17a,AB24} (\cite{FK18,FIK20,LMS22}); for the ensemble of induced GinOE matrices \cite[\S 9.2]{BF25};
and most recently for asymmetric Wishart matrices \cite{BN25}.

Knowledge of certain $p_{N,M}^{\rm r}$ appear in applications. Thus in \cite{La13}  the problem in quantum information of
quantifying when two-qubits, chosen from a uniform distribution on the unit 3--sphere, 
 are an optimal pair was considered. The condition for the latter
 as  particular inequalities for certain weighted inner products between the qubits \cite{SRL11}. In \cite{La13} it was shown that the probability of this event is equal to the probability $p_{2,0}^{\rm r}$ for the ensemble of the product of two independent $2 \times 2$
GinOE matrices.
As the next example,
consider the tensor structure $\mathcal A = (a_{ijk}) \in
\mathbb R^{p \times p \times 2}$, represented as the column vector
${\rm vec } \, \mathcal A \in \mathbb R^{4 p^2}$. 
A problem of interest 
is to find matrices
$U = [ \vec{u}_1 \cdots \vec{u}_R] \in \mathbb R^{p \times R}$,
$V = [ \vec{v}_1 \cdots \vec{v}_R] \in \mathbb R^{p \times R}$,
$W = [ \vec{w}_1 \cdots \vec{w}_R] \in \mathbb R^{2 \times R}$, where $R$ referred to as the rank, such that the decomposition
$
{\rm vec} \, \mathcal A = \sum_{r=1}^R \vec{w}_r \otimes \vec{v}_r \otimes \vec{u}_r
$
holds for $R$ as small as possible \cite{KB09}. It is known
\cite{tB91} that if both $[a_{ij1}] =: X_1 \in \mathbb R^{p \times p}$ and
$[a_{ij2}] =: X_2 \in \mathbb R^{p \times p}$ are random matrices, with entries chosen from a continuous
distribution, then $R = p$ if all the eigenvalues of $X_1^{-1} X_2$ are real, and $R=p+1$ otherwise. Thus for $X_1,X_2$ GinOE matrices, the respective probabilities are $p_{N,N}^{\rm r}$ and $1 - p_{N,N}^{\rm r}$, with $p_{N,k}^{\rm r}$ now referring to the probability of $k$ real eigenvalues in the spherical ensemble. The most recent application, found in \cite{PS18}, shows that $p_{N,0}^{\rm r}$ for $N$ large
in the case of a single row and column truncation of Haar distributed real orthogonal matrix, relates to the persistance exponent for two-dimensional diffusion with random initial conditions --- for more on the asymptotics see \cite{GP19,FTZ22}.

Beyond the mean number of real eigenvalue for large $N$, a natural statistical question is their
fluctuation. In the case of the GinOE, the result that the variance for large $N$ is proportional to the mean was established in \cite{FN07}. In \cite{FM11} this was shown to remain true for the real spherical ensemble, and moreoever a local central limit theorem (local CLT) quantifying the probability distribution in the neighbourhood of its mean was deduced --- see too \cite[\S 9.3]{BF25} and \S \ref{S3.1} below. A closely related CLT, applying more generally to a broad class of (not necessarily smooth) linear statistics of the real eigenvalues, has been given in the case of the GinOE in \cite{Si17,FS23}. Results along this latter  line for the elliptic GinOE have been given in \cite{Fo24,BMS25}. By way of applications of fluctuation formulas for linear statistics, and distribution functions of the real eigenvalues generally for elliptic GinOE, one draws attention to the role they play in the study of complexity of random landscapes \cite{Fy16,BFK21}. 

The aim of the present paper is, for the real spherical ensemble, to specify asymptotic results for $p_{N,M}^{\rm r}$ beyond the regime of the validity of the local CLT.
Based on a products of gamma function evaluation in the case $k=N$ from \cite{FF11}, a result of this type has previously been obtained in \cite[Eq.~(29)]{BF11a},
where it was shown
\begin{equation}\label{(16b)}
\log p_{N,N}^{\rm r} = {N^2 \over 4} - {N^2 \over 2} \log 2 + {1 \over 12} \log N - {1 \over 12} -
\zeta'(-1) + {\rm O}\Big ( {1 \over N} \Big ).
\end{equation}
Our first result extends the O$(N^2)$ 
(i.e.~leading order  term herein) to specify the leading order  term for all $p_{N,M}^{\rm r}$ with
$M/N = \alpha$ ($0 < \alpha  \le 1$).

\begin{prop}\label{P1}
(Conditional on the validity of an electrostatics hypothesis.)
Let $\alpha=M/N$. For large $N$ and with $\alpha \ne 0$ fixed we have
\begin{multline}\label{(17x)}
\log p_{N,M}^{\rm r} \sim
- \frac{N^2}{8} \bigg (-2 \alpha +(\alpha-1) \Big ( (1-\alpha) \log (1 - \alpha)-(\alpha+1) \log (1+\alpha) \Big )\\ +2 \alpha (\alpha+1)
   \log (1 + \alpha) \bigg ).
\end{multline}
(Here the case $\alpha = 1$ is to be interpreted as the limit $\alpha \to 1^-$; note then agreement with the leading order in (\ref{(16b)}).)
\end{prop}

This is established in Section \ref{S2}. We then proceed in Section \ref{S3} to analyse the leading order asymptotics of $\log p_{N,M}^{\rm r}$ for $M$ proportional to 
$\sqrt{N}$, which is the intermediate deviation regime. The analysis here relies on determining the leading asymptotics of the generating function for the probabilities $\{p_{N,M}^{\rm r} \}$, the applying steepest descents to the corresponding contour integral formula. Matching of the left and right tails with the asymptotic form in the right tail of the large deviation result, and for the known local central limit regime
result respectively, is exhibited. Also obtained is the leading order asymptotics of
the probability of no real eigenvalues $p_{N,0}^{\rm r}$.

\section{Electrostatic energy and proof of Proposition \ref{P1}}\label{S2}
\subsection{An hypothesis linking $p_{N,M}^{\rm r}$ and an electrostatic energy}

To appreciate the inter-relation of the heading of this subsection, it is helpful to first revise the situation with the complex (rather than real) spherical ensemble \cite[\S 15.6]{Fo10}. This consists of random matrices $A B^{-1}$, where $A,B$ are independent GinUE (standard complex Gaussian) random matrices. From \cite{Kr09} the corresponding eigenvalue PDF is proportional to 
\begin{equation}\label{(1E)}
\prod_{l=1}^N {1 \over (1 + | z_l|^2)^{N+1}}
\prod_{1 \le j < k \le N} |z_k - z_j|^2.
\end{equation}
Applying the stereographic projection $z_l = e^{i \phi_l} \sin {\theta_l \over 2}$ from a point on the plane, to a point $\vec{r}_l$ on the sphere of radius $1/2$ in $\mathbb R^3$ centred at the origin with polar, azimuthal angle pair $(\theta_l, \phi_l)$, transforms (\ref{(1E)}) to the functional form
\begin{equation}\label{(2E)}
\prod_{1 \le j < k \le N} ||\vec{r}_k - \vec{r}_j||^2.
\end{equation}
The functional form (\ref{(2E)}) is, up to proportionality, the Boltzmann factor for the two-dimensional one-component plasma (Coulomb gas) on the sphere at inverse temperature $\beta = 2$.
The exact solution of this model in the sense of the free energy and correlation functions was given in \cite{Ca81}. Very recently various integrable properties have been shown to persist in the presence of particular external potentials, typically related to point insertions;
see \cite{BKSY25}, \cite{BFKL25}, \cite{BFL25}.

The exact functional form of the eigenvalue PDF for the real spherical ensemble has been calculated in \cite{FM11}. This functional form is conditional on the number of real eigenvalues. An approximate identification with a conditioned Coulomb gas, placing charges on the equator to correspond to the real eigenvalues, and charges in the upper (lower) hemisphere to correspond to the complex eigenvalues in the (upper) lower half complex plane, is given in \cite{Fo16}.
Since the complex eigenvalues occur in complex conjugate pairs, the charges in the lower half sphere are the mirror image of those in the upper half plane. The correspondence is not exact as there are additional one-body terms which require expanding at large distances before the identification. 

This is similarly true in relating the conditional eigenvalue PDF for GinOE to a Coulomb gas system \cite{Fo16}. Notwithstanding this complication, it was shown in \cite{GPTW16} in the case of GinOE that the approximate Coulomb gas analogy is exact at leading order with respect to the asymptotic form of the configuration integral determining the probabilities $p_{N,M}^{\rm r}$. Consequently an expression for 
\begin{equation}\label{ST}
 \lim_{N \to \infty}\bigg (- {1 \over N^2}
\log p_{N,M} \bigg ) \bigg |_{M/N = \alpha}
\end{equation}
in terms of the electrostatic energy of a minimising conditioned continuous charge configuration was given. The requirement of the latter was that there be a charge density on the segment $[-1,1]$ of the real line, total charge $\alpha$, and a charge density in the unit disk, symmetrical about the real axis of total charge $1-\alpha$.
However, apart from the case $\alpha = 0$, it was not possible to determine this energy and thus evaluate (\ref{ST}), in keeping with the absence of knowledge of the minimising charge densities. 

The starting point of a our study in relation to the asymptotics of $ p_{N,M}  |_{M/N = \alpha}$ for the real spherical ensemble is to hypothesize the validity of the analogous electrostatic formulation of (\ref{ST}) as demonstrated in \cite{GPTW16} in the case of GinOE. To attempt to give a proof would change the emphasis of this work from exact asymptotic formulas to the technical details associated with large deviations analysis, which we want to avoid\footnote{Historical precedence for following this route can be found in the pioneering work of Dyson
\cite{Dy62e} in relation to the computation of probability of a gap of length $\alpha$ in the bulk scaled circular $\beta$ ensemble.}.
Specifically, by analogy with the findings of \cite{GPTW16}, we hypothesize (\ref{ST}) as equal to the minimum energy, $E_\alpha$ say, 
of a specific electrostatics problem.
Thus,
with respect to the Coulomb pair potential on the sphere, $E_\alpha$ is the minumum possible electrostatic
energy resulting from distributing charge such that the total charge on the equator is $\alpha/\pi$, and the charge distribution on the body of the sphere is symmetrical about the equator of total charge density $(1-\alpha)/ \pi$
(the factors of $1/\pi$ come from the fact that for a sphere of radius $1/2$ the equator has length $\pi$, and the area of the sphere also equals $\pi$).
As a normalisation giving $E_\alpha = 0$ when $\alpha = 0$, it is supposed too that there is a  neutralising background charge density $-1/\pi$, uniform on the sphere. 

In distinction to the case of the GinOE, the charge distribution for this minimisation is simple to identify. 
First, by the rotational symmetry of the problem about the polar axis, the charge density both on the equator and the body of the sphere must be independent of the azimuthal angle $\phi$.
The same symmetry and the repulsion of the Coulomb potential between like charges implies that the support of the positive charge density on the body of the sphere consists of a spherical cap about the north (south) pole with polar angle $\theta \in [0,\theta_0]$ ($\theta \in [\pi - \theta_0, \pi]$). 

Furthermore, so that there is no electrtic field produced by this charge density when superimposed on the uniform negative background charge density, the charge density in each spherical cap must equal $1/\pi$. On the other hand, the total charge in a spherical cap is $(1-\alpha)/2$, and so taking the one about the north pole for definiteness (let its area be denoted $A_{[0,\theta_0]}$) we have
\begin{equation}\label{(1)}
{1 - \alpha \over 2} = {1 \over \pi} A_{[0,\theta_0]}.
\end{equation}
Now, for a sphere of radius ${1 / 2}$ we know the area differential is $(1 / 4) \sin \theta d \theta d \phi$; see e.g.~\cite[proof of Prop.~2.1]{FF11}. Hence
\begin{equation}\label{(2)}
A_{[0,\theta_0]} = {\pi \over 2} \int_0^{\theta_0}
\sin \theta \, d \theta = {\pi \over 2} (1 - \cos \theta_0) = \pi \sin^2 {\theta_0 \over 2} =
\pi \Big ( 1 - \cos^2 {\theta_0 \over 2} \Big ).
\end{equation}
Substituting the second equality in (\ref{(1)}) gives
\begin{equation}\label{(3)}
 \cos \theta_0 = \alpha.
\end{equation}
% We similarly require that about the spherical cap about the south pole, ending at $\theta = \pi - \theta_0$, there is
% a neutralising uniform charge density $+ {N \over \pi}$, such that the total charge in this cap is equal to ${N - M \over 2}$. Note that this requirement is consistent with the formula (\ref{(3)}) for $\theta_0$. 

\subsection{Computation of the electrostatic energy}
We are now in a position to carry out the computation required to deduce (\ref{(17x)}) of Proposition \ref{P1}.
Due to the neutralising spherical caps 
about the north and south poles, the equivalent charge distribution is a uniform negative background charge density $- 1 / \pi$ in the spherical annulus about the equator with polar angles in the interval $(\theta_0, \pi - \theta_0)$ (total charge in this region is $-\alpha$), and the uniform positive charge density ${\alpha/ \pi}$ on the equator. For  the total electrostatic energy we have
\begin{equation}\label{(4)}
E_{\alpha} = E^{\rm a} +
E^{\rm e} + E^{\rm e/a},
\end{equation}
(here $E^{\rm e/a}$ is the electrostatic energy between the charge distribution on the equator and the spherical annulus, while $E^{\rm a},E^{\rm e}$ are self explanatory). 

Let $V_{\rm a}$ denote the potential energy due a point charge in the annulus interacting with the smeared out uniform negative background in the annulus. Proceeding as in \cite{FF11}, we compute this according to
\begin{equation}\label{(4.1)}
V_{\rm a} = V_{\rm a/s} - V_{\rm a/caps},
\end{equation}
where $V_{\rm a/s}$ ($V_{\rm a/caps}$)
is the potential energy due to a point charge in the annulus interacting with the smeared out uniform negative background
over the entire sphere (the two spherical caps).

From \cite[Eq.~(2.10) with $Q=q=0$]{FF11}, the potential energy due to a point positive unit charge anywhere on the sphere and the smeared out negative background is $-{1 \over 2}$, and thus
\begin{equation}\label{(4.2)}
V_{\rm a/s} =  -{1 \over 2}.
\end{equation}

% The charge density of the smeared out background is $- {N \over \pi}$. We have 
% \begin{equation}\label{(5)}
% E^{\rm sphere} = {1 \over 2} V_s \Big ( - {N \over \pi} \Big ) \pi{N - M \over N} = {N(N-M) \over 4}.
% \end{equation}
% Here the first factor of ${1 \over 2}$ is due to the double counting, while the last factor of $\pi {N - M \over N}$ is the area of the annulus as calculated in
% (\ref{(2)}) and (\ref{(3)}).
% The value (\ref{(5)}) is the first term in \cite[Eq.~(2.18)]{FF11}.

From \cite[Eq.~(2.13)]{FF11}, the potential $V_{a/\rm Ncap}$ due to a point positive unit charge at polar angle $\theta'$ in the spherical annulus about the equator, and the neutralising uniform negative background charge in the spherical cap about the north pole is
\begin{equation}\label{(6)}
V_{{\rm a/Ncap}} = 
 {1 \over \pi} {\pi \over 2}
\int_0^{\theta_0} \sin \theta \Big (
\log \cos {\theta \over 2} +
\log \sin {\theta' \over 2} \Big ) \, d \theta.
\end{equation}
Here $ {1 / \pi}$ is the (positive) neutralising charge density in the spherical cap, while $\pi / 2$ is the constant factor in the Jacobian coming from integrating  $(1/4) \sin \phi$ (with $R= {1 / 2}$) over the azimuthal angle $\phi$. The factor of $\sin \theta$ in the integral is part of the Jacobian, while the factor 
$\Big (
\log \cos {\theta \over 2} +
\log \sin {\theta' \over 2} \Big )$ comes from a rewrite of minus the Coulomb potential on the sphere between unit charges at $(\theta,\phi)$
and $(\theta',\phi')$, integrated over $\phi$ with the assumption that $\theta'$ is in the spherical annulus.

Making use of the (\ref{(3)}) we compute
$$
\int_0^{\theta_0} \sin \theta \, d \theta = 1 - \cos \theta_0 = 1 - \alpha.
$$
Also
\begin{multline}
\int_0^{\theta_0} \sin \theta \Big (
\log \cos {\theta \over 2} \Big ) \, d \theta = - \sin^2 {\theta_0 \over 2} - 
\cos^2 {\theta_0 \over 2} \log \cos^2 {\theta_0 \over 2} \\
= - {1 \over 2} \Big ( 1 - \alpha \Big ) -
{1 \over 2} \Big ( 1 + \alpha \Big )
\log {1 \over 2} \Big ( 1 + \alpha \Big ),
\end{multline}
where the first equality follows by relating it to the anti-derivative of $x \log x$, and the second makes use of formulas implied by 
(\ref{(2)}) and (\ref{(3)}).
Hence
\begin{multline}\label{(8)}
V_{{\rm a/Ncap}} =  {1 \over 2} \bigg (
- {1 \over 2}  ( 1 - \alpha  )-{1 \over 2}  ( 1 + \alpha  )
\log {1 \over 2}  ( 1 + \alpha  )\\
+  \Big ( \log \sin {\theta' \over 2} \Big )
 ( 1 - \alpha  ) \bigg ).
\end{multline}
In relation to 
the potential $V_{a/\rm Scap}$ due to a point positive unit charge at polar angle $\theta'$ in the spherical annulus about the equator, and the  uniform negative background charge in the spherical cap about the south pole,
we simply replace $\theta'$ in (\ref{(8)}) by $\pi - \theta'$ to get
\begin{multline}\label{(9)}
V_{{\rm a/Scap}} =  {1 \over 2} \bigg (
- {1 \over 2}  ( 1 - \alpha  )-{1 \over 2} ( 1 + \alpha  )
\log {1 \over 2}  ( 1 + \alpha  )\\
+ \Big ( \log \cos {\theta' \over 2} \Big )
 ( 1 - \alpha  ) \bigg ).
\end{multline}

Substituting the sum of (\ref{(8)}) and
(\ref{(9)}) for $V_{\rm a/caps}$ in
(\ref{(4.1)}), and substituting too  (\ref{(4.2)}) we have
\begin{multline}\label{(9.1)}
V_{{\rm a}} = -{1 \over 2}   -{1 \over 2}
\bigg ( -   ( 1 - \alpha  )-  ( 1 + \alpha  ) 
\log {1 \over 2}  ( 1 + \alpha  ) \\
+ \Big ( \log \cos {\theta' \over 2}  + \log \sin {\theta' \over 2} \Big  )
 ( 1 - \alpha  ) \bigg ) \\
 =  -{1 \over 2}
\bigg (  \alpha -  ( 1 + \alpha ) 
\log {1 \over 2}  ( 1 + \alpha  ) 
+ \Big ( \log \cos {\theta' \over 2}  + \log \sin {\theta' \over 2} \Big  )
 ( 1 - \alpha  ) \bigg ).
\end{multline}
The task now is to compute
\begin{equation}\label{(10)}
E^{\rm a} =
 {1 \over 2} \Big ( - {1 \over \pi} \Big ) {\pi \over 2} 
\int_{\theta_0}^{\pi - \theta_0}  
V_{{\rm a}} \sin \theta  \, d \theta.
\end{equation}
The first factor of ${1 / 2}$ comes from the double counting from the implied double integral. Also,
similar to (\ref{(6)}), the factor
$ -{1 / \pi}$ is the (negative) background charge density in the spherical cap, while the factor ${\pi / 2}$ is the constant factor in the Jacobian coming from integrating  $(1/4) d \phi$  over the azimuthal angle $\phi$. The factor of $\sin \theta$ in the integral is part of the Jacobian.

To compute (\ref{(10)}), we require the integrals
$$
 \int_{\theta_0}^{\pi - \theta_0} \sin \theta  \, d \theta =
 - \Big ( \cos(\pi - \theta_0) - \cos \theta_0 \Big )= 2 \cos \theta_0 = 2 \alpha
$$
and
\begin{multline*}
 \int_{\theta_0}^{\pi - \theta_0} \sin \theta 
 \Big ( \log \sin {\theta \over 2} +
 \log \cos {\theta \over 2} \Big )
 \, d \theta =
 \int_{\theta_0}^{\pi - \theta_0} \sin \theta  \log \Big ( {1 \over 2} \sin \theta \Big ) \, d \theta \\ = 
 2 \int_{\theta_0}^{\pi/2} \sin \theta  \log \Big ( {1 \over 2} \sin \theta \Big )
 \, d \theta
= 2 \bigg (
 \log \cot {\theta_0 \over 2} +
 \cos \theta_0 \Big ( -1 + \log {\sin \theta_0 \over 2} \Big ) \bigg ) \\
 = 2 \bigg \{ 
 {1 \over 2} \log {1 + \alpha \over 1 - \alpha} +
 \alpha \bigg ( - 1 +
 \log \Big ( {1 \over 2}  ( 1 - {\alpha}^2  )^{1/2} \Big ) \bigg ) \bigg \}.
 \end{multline*}
 Hence
 \begin{multline}\label{(10a)}
E^{\rm a} =  {1 \over 4} \bigg [
 \bigg (
   \alpha - \Big ( 1 + \alpha \Big )
\log {1 \over 2} \Big ( 1 + \alpha \Big ) \bigg ) \alpha \\
+  ( 1 - \alpha  )  
\bigg \{ 
 {1 \over 2} \log {1 + \alpha \over 1 - \alpha} +
 \alpha \bigg ( - 1 +
 \log \bigg ( {1 \over 2} \Big ( 1 - \alpha^2 \Big )^{1/2} \Big ) \bigg ) \bigg \} \bigg ].
 \end{multline}

 We consider next the computation of $E^{\rm equator}$. For particles on the equator separated by azimuthal angle $\phi$, we use the fact that the pair potential is $-\log \sin \phi/2$. Hence the potential energy of a unit positive point charge on the equator, and a positive continuous charge density
 ${\alpha / \pi}$ on the equator is
 \begin{equation}\label{(13)}
 V_{\rm e} = - {\alpha \over \pi} {1 \over 2}
 \int_0^{2 \pi} \log \sin \phi/2 \, d \phi = - {\alpha \over \pi} {1 \over 2} 
 \Big ( -2 \pi \log 2 \Big ) = \alpha \log 2.
 \end{equation}
Here, since the radius of the equator is $R= {1 / 2}$, we have used the fact that the length differential is $(1 / 2) d \phi$. From this we calculate
 \begin{equation}\label{(14)}
 E^{\rm e} = {1 \over 2} {\alpha \over \pi} {1 \over 2}
 \int_0^{2 \pi}  V_{\rm e} \, d \phi
 = {\alpha^2 \over 2} \log 2.
\end{equation} 

It remains to compute $E^{\rm e/a}$. We have
\begin{equation}\label{(15)}
 E^{\rm e/a} = {\alpha \over \pi} {1 \over 2}
 \int_0^{2 \pi}  V_{\rm a} \Big |_{\theta = \pi/2} \, d \phi .
 \end{equation}
The factor ${\alpha \over \pi}$  is the charge density on the equator of the sphere (which is positive) and ${1 \over 2} d \phi$ is the length differential.
Making use of (\ref{(9.1)}) we therefore have
\begin{equation}\label{(15a)}
 E^{\rm e/a} = \alpha V_{\rm a} \Big |_{\theta = \pi/2} =
 -{\alpha \over 2}
\bigg (  \alpha -  ( 1 + \alpha  ) 
\log {1 \over 2}  ( 1 + \alpha  ) 
+   ( 1 - \alpha  ) \log  {1 \over 2} 
 \bigg ).
\end{equation}
Combining with (\ref{(14)}) then gives
\begin{equation}\label{(16)}
E^{\rm e}+ E^{\rm e/a} = 
{\alpha \over 2} \bigg ( - \alpha (1 + \log 2 )
+ (1+\alpha) \log (1 + \alpha) \bigg ).
\end{equation}
Adding this to (\ref{(10a)}) according to (\ref{(4)}) and recalling our fundamental hypothesis that (\ref{ST}) is equal to $E_\alpha$
we obtain the stated result (\ref{(17x)}).
\hfill $\square$
%  \begin{multline}\label{(15a)}
% E_{N,M} = \frac{1}{8} \bigg (-2 M N+(M-N) \Big ( (N-M) \log {(N-M) \over N}-(M+N) \log {(M+N) \over N}\Big )\\ +2 M (M+N)
%    \log \frac{(M+N)}{N}\bigg ).
% \end{multline} 
% Setting $\alpha = M/N$ this reads
%  \begin{multline}\label{(15a+)}
% E_{N,M} = \frac{N^2}{8} \bigg (-2 \alpha +(\alpha-1) \Big ( (1-\alpha) \log (1 - \alpha)-(\alpha+1) \log (1+\alpha) \Big )\\ +2 \alpha (\alpha+1)
%    \log (1 + \alpha) \bigg )
% \end{multline}

\begin{remark}\label{R1}
It follows from (\ref{(17x)}) that for small $\alpha$ (which is the neighbourhood of the minimum of $E_{\alpha}$ as a function of $\alpha$ for $1 > \alpha \ge 0$),
\begin{equation}\label{15ab}
E_{\alpha} \sim \frac{N^2}{8} \bigg ( \frac{2 \alpha^3}{3}+\frac{\alpha^5}{15}+{\rm O}\left(\alpha^7\right) \bigg ).
\end{equation}
This leading cubic vanishing with respect to $\alpha$ gives that $E_\alpha$ has order unity behaviour for $\alpha$ of order $N^{-2/3}$. On the other hand, as to be revised in \S \ref{S3.1} below, the region of validity of a local CLT is of
order $N^{-3/4}$. Such a discrepancy, first noted in the context eigenvalue fluctuations in a centered disk region for the Ginibre ensemble of standard complex Gaussian matrices in \cite{L+19},  was identified in this latter work as an indicator of an intermediate regime linking the region of large fluctuations to that of the local CLT. Such an intermediate regime in the present context is discussed in \S \ref{S3.2} below.

% The significance of such a cubic form, in relation to an intermediate scaling regime between that of large deviations, and the regime of the validity of a central limit theorem, is discussed in \cite{L+19}.
\end{remark}

\subsection{Comparison with exact numerics}

% Consider the spherical ensemble of real Gaussian random matrices, formed out of matrices $AB^{-1}$, where $A,B$ are independent $N \times N$ GinOE matrices; see e.g.~\cite{BF25}. Let $p_{N,M}^{\rm r}$ denote the probability that there are $M$ real eigenvalues (for this to be nonzero it is required that $M$ have the same parity as $N$). The large deviation calculation of \cite{GPTW16} gives that to leading order in $N$, and with $M/N$ fixed ($M,N$ having the same parity)
% \begin{equation}\label{(17)}
% \log p_{N,M}^{\rm r} \sim - E_{N,M}.
% \end{equation}

% This can be checked for $M=N$. Then (\ref{(15a)}) simplifies to
% \begin{equation}\label{(16a)}
% E_{N,N} = - {N^2 \over 4} + {N^2 \over 2} \log 2,
% \end{equation}
% which substituted in (\ref{(17)}) implies
% \begin{equation}\label{(17)}
% \log p_{N,M}^{\rm r} \sim {N^2 \over 4} - {N^2 \over 2} \log 2.
% \end{equation}
% On the other hand, the formula in \cite{FM11} for
% $p_{N,M}^{\rm r}$, general $M$, allows for $M=N$ the computation of the large $N$ asymptotic formula
%  \cite[Eq.~(29)]{BF11a}
% \begin{equation}\label{(16b)}
% \log p_{N,N}^{\rm r} = {N^2 \over 4} - {N^2 \over 2} \log 2 + {1 \over 12} \log N - {1 \over 12} -
% \zeta'(-1) + {\rm O}\Big ( {1 \over N} \Big ).
% \end{equation}

The exact formula in \cite{FM11} for
$p_{N,M}^{\rm r}$, general $M$, referred to above relates to the corresponding generating function
\begin{eqnarray}\label{3'}
 Z_N(\xi)=\sum_{M=0}^{N}{}^* \, \xi^M p_{N,M}^{\rm r},
\end{eqnarray}
where the asterisk indicates that the sum over $M$ is restricted to values with the same parity as $N$.
Taking $N$ even for convenience, we have from \cite{FM11} the product evaluation
\begin{eqnarray}\label{3a}
\nonumber Z_N(\xi)&=&\frac{(-1)^{(N/2)(N/2-1)/2}}{2^{N(N-1)/2}}\Gamma((N+1)/2)^{N/2}\Gamma(N/2+1)^{N/2}\\
\label{eqn:Z_N} &&\times\prod_{s=1}^{N}\frac{1}{\Gamma(s/2)^2}\prod_{l=0}^{N/2-1}(\xi^2 \alpha_{l} + \beta_{l}),
\end{eqnarray}
where
\begin{eqnarray}\label{3'a}
\nonumber \alpha_{l}&=&\frac{2\pi}{N-1-4l}\frac{\Gamma((N+1)/2)}{\Gamma(N/2+1)},\\
\beta_{l}&=&\frac{2\sqrt{\pi}}{N-1-4l}\left( 2^N\frac{\Gamma(2l+1)\Gamma(N-2l)}{\Gamma(N+1)}-\sqrt{\pi}\frac{\Gamma((N+1)/2)}{\Gamma(N/2+1)}\right).
\end{eqnarray}
For numerical results for a specific $N$, computer algebra can be used to expand (\ref{3a}) in the form of (\ref{3'}) to thus read off the probabilities $\{p_{N,M}^{\rm r} \}_{M=0,2,\dots,N}$. This allows for a check on (\ref{(17x)}). For example, using (\ref{3a}) we calculate
$$
 p_{60,30}^{\rm r} = 3.562969809273*10^{-17} \quad \implies \log p_{60,30}^{\rm r} = -37.87335217\cdots
$$
while (\ref{(17x)}) gives
$$
\log p_{60,30}^{\rm r} \approx - 38.5125\cdots
$$

\section{Intermediate deviation and probability of no real eigenvalues}\label{S3}

\subsection{Region about the mean}\label{S3.1}
A feature of the generating function $Z_N(\alpha)$ (\ref{3a}) in factorised form is that all the zeros as a function of $\xi^2$ are negative and real.
As noted in \cite{FM11}, according to \cite{Be73}, this is a sufficient condition that the probabilities $\{ p_{N,M}^{\rm r} \}$ satisfy a local CLT about the mean
\begin{equation}\label{M}
 \langle M \rangle := \mu_N \sim \sqrt{\pi N/2}.   
\end{equation}
Thus, with knowledge too of the variance
$\sigma_N^2 := \langle M^2 \rangle - \langle M \rangle^2 \sim (2 - \sqrt{2}) \mu_N$ also computed in \cite{FM11}, we read off from \cite[Prop.~3.5]{FM11}
the asymptotic form for $M$ in the neighbourhood of the mean,
\begin{equation}\label{14a}
p_{N,M}^{\rm r} \Big |_{M/N = \alpha} \sim {1 \over
\sqrt{\pi c N^{1/2}}} e^{-N^{3/2} ( \alpha - \mu_N/N)^2/c}, \quad c:= \sqrt{2 \pi} (2 - \sqrt{2}).
\end{equation}
Note that for the exponent in (\ref{14a}) to be of order unity as required for the regime of the validity of the local CLT, it is required that
\begin{equation}\label{14b}
\alpha - \mu_N/N \sim {\rm O}(N^{-3/4}).
\end{equation}
With this condition, a numerical check on the functional form (\ref{14a}) carried out using 
(\ref{3a}) has been displayed in \cite[Fig.~2]{FF11}.

As already pointed out in Remark \ref{R1}, the exponent $-3/4$ in (\ref{14b}) is distinct from the exponent $-2/3$ associated with the right tail of the large deviation asymptotic form. Also noted in Remark \ref{R1} was the identification of such a circumstance in a seemingly distinct problem in random matrix theory \cite{L+19}, where it was shown to be a signature of an intermediate deviation regime. Moreover a strategy to analyse the latter was devised in the work \cite{L+19}, which we will show suffices for the present problem in the next subsection.

\subsection{Intermediate deviation}\label{S3.2}
According to (\ref{3'})
\begin{equation}\label{2.23a}
p_{N,M}^{\rm r} = [\xi^M] Z_N(\xi) = 
{1 \over 2 \pi i } \oint
{Z_N(\xi) \over \xi^{M+1}} \, d \xi.
\end{equation}
Here $[\xi^M]$ denotes the coefficient of $\xi^M$ in the following function, and the integral in the contour is required to be simple, closed, and to encircle the origin.

 The leading large $N$ asymptotic form of $Z_N(\xi)$ can be calculated, in the situation that $\xi$ is fixed.

 \begin{prop}\label{P2}
For large $N$ we have
\begin{equation}\label{px2}
\log Z_N(\xi) \sim { \sqrt{N \over 2}}
\int_0^\infty \log \Big (
1- (1-\xi^2) e^{-  t^2} \Big ) \, dt.
\end{equation}
 \end{prop}

\begin{proof}
We notice from   (\ref{3'}) that
$$
\beta_l = \beta_l' - \alpha_l, \qquad\beta_l' =
\frac{2\sqrt{\pi}}{N-1-4l}2^N\frac{\Gamma(2l+1)\Gamma(N-2l)}{\Gamma(N+1)}.
$$
Furthermore, since $Z_N(\xi) |_{\xi = 1} =1$, and the fact that $Z_N(\xi) |_{\xi = 1}$ is given by   (\ref{3a}) with the factors $\xi^2 \alpha_l+ \beta_l$ replaced by $\beta_l'$, a possible rewrite of (\ref{3a}) is that
\begin{multline}\label{3av}
 Z_N(\xi) =  \prod_{l=0}^{N/2-1}\Big (1 - (1-\xi^2) {\alpha_{l} \over \beta_l'}  \Big ) \\=
 \prod_{l=0}^{N/2 - 1} \bigg (1 - (1-\xi^2)
\frac{\sqrt{\pi}}{2^N}
\frac{\Gamma(N+1)}{\Gamma(2l+1)\Gamma(N-2l)}
{ \Gamma((N+1)/2) \over \Gamma(N/2+1) } \bigg ).
\end{multline}

For large $N$, $\Gamma((N+1)/2)/\Gamma(N/2+1) \sim N^{-1/2}$ while
$$
\frac{\sqrt{\pi}}{2^N}
\frac{\Gamma(N+1)}{\Gamma(2l+1)\Gamma(N-2l)} =
\frac{\sqrt{\pi}}{2^N} \binom{N}{2l} (N-2l) \sim
2 e^{-(N-4l)^2/2N} (1 - 2l/N), 
$$
with the asymptotic form (which is valid for $(N - 4l) = {\rm o}(N^{2/3})$) following from the well documented local CLT as applied to the binomial coefficients; see e.g.~\cite{WiB}.
Hence
$$
\bigg (1 - (1-\xi^2)
\frac{\sqrt{\pi}}{2^N}
\frac{\Gamma(N+1)}{\Gamma(2l+1)\Gamma(N-2l)}
{ \Gamma((N+1)/2) \over \Gamma(N/2+1) } \bigg ) \sim
1 - (1-\xi^2) e^{-k^2/2N},
$$
where $k:=N - 4l$. We therefore have
\begin{equation}
   \log Z_N(\xi) \sim \sum_{l=0}^{N/2-1} \log \Big (
1- (1-\xi^2) e^{-k^2/2N} \Big )  
\end{equation}
and (\ref{px2}) follows by recognising this as a Riemann sum.
\end{proof}

Use of (\ref{px2}) in (\ref{2.23a}) allows for the computation of the leading order asymptotics in the regime $M = {\rm O}(N^{1/2})$.

\begin{prop}
Let $M = xN^{1/2}$, $x > 0$ and chosen so that $M$ is an integer. Define
\begin{equation}\label{c1}
\chi(\mu) = {1 \over \sqrt{2}} 
\int_0^\infty \log \Big (
1- (1-e^{-2 \mu}) e^{-  t^2} \Big ) \, dt.
\end{equation}
We have
\begin{equation}\label{c2}
p_{N,M}^{\rm r} \Big |_{M = xN^{1/2}} \sim
\exp \Big ( \sqrt{N} \mathop{\min}\limits_{\mu \in \mathbb R} ( x \mu + \chi(\mu) ) \Big ).
\end{equation}
\end{prop}

\begin{proof}
As indicated at the end of the previous subsection, our strategy is that first introduced in \cite{L+19}
in the context of eigenvalue fluctuations in a disk for the Ginibre ensemble.
We begin by changing variables $\xi = e^{-\mu}$ in the contour integral of (\ref{2.23a}) to obtain the equivalent integral form
\begin{equation}\label{2.23b}
p_{N,M}^{\rm r} = \int_I e^{\mu M} Z_N(e^{-\mu}) \, d\mu,
\end{equation}
where $I$ is a simple contour in the complex plane starting at $\mu_0$ and finishing at $\mu_1$, with
Im$\, \mu_0 = - \pi$, Im$\, \mu_1 =  \pi$ and
Re$\,(\mu_0) = {\rm Re} \, (\mu_1)$. 
Denote $1/\sqrt{2}$ times the integral in (\ref{px2}) with the substitution $\xi = e^{-\mu}$ as $\chi(\mu)$ as in (\ref{c1}).
Substituting in (\ref{2.23b}) then gives that for large $N$
\begin{equation}\label{2.23c}
p_{N,M}^{\rm r} \Big |_{M = xN^{1/2}}
\sim \int_I e^{ \sqrt{N} (x \mu + \chi(\mu))} \, d\mu.
\end{equation}

Straightforward considerations show that there is a unique stationary point of $\mu + \chi(\mu)$ as appearing in the exponent (\ref{2.23c}) on the real axis. One can check that too that this corresponds to a local minimum, and that the function value is then negative. We choose $\mu_0$ and $\mu_1$ as introduced below (\ref{2.23c}) so that $I$ passes through the stationary point.
Moreover the fact that it is a local minimum is compatible with choosing $I$ as parallel to the negative real axis, which then corresponds to passing through the stationary point in the direction of steepest descent. The leading order asymptotic form (\ref{c2}) follows.
\end{proof}    

The extremities of the intermediate deviation result (\ref{c2}) are $x \to \infty$ and $x \to 0$. These limits relate to the large deviation and central limit regimes respectively. By following the working of \cite[Appendix B]{L+19} carried out in relation to the distribution of the counting function for the number of eigenvalues inside a disked shape region in the Ginibre ensemble, it can be shown that there is a precise matching with the tails in each case.

Consider first the limiting case $x \to \infty$. Here we expect this to relate to the form of $\chi(\mu)$ for $\mu \to - \infty$. In the latter limit it follows from (\ref{c1}) that
$$
\chi(\mu) \mathop{\sim}\limits_{\mu \to -\infty} {1 \over \sqrt{2}}
\int_0^{\sqrt{2|\mu|}}
(2 |\mu| - t^2) \, dt = {4 \over 3}|\mu|^{3/2}.
$$
Substituting this for $\chi(\mu)$ in the exponent of (\ref{c2}) gives that the minimum occurs for $\mu^* = - 2 x^2$. Now evaluating this minimum shows that for $x$ large the probability has the leading order form 
\begin{equation}\label{2.33c}
p_{N,M}^{\rm r} |_{M=xN^{1/2}} \sim e^{-\sqrt{N} x^3/12}. 
\end{equation}
On the other hand, in relation to the large deviation regime, $M = \alpha N$ and thus $x = \alpha N^{1/2}$. This allows the RHS of (\ref{2.33c}) to be identified as $e^{- \alpha^3 N^2/12}$, which is precisely the large deviation result (\ref{(17x)}) with the substitution on the RHS given by minus the leading order small $\alpha$ form of
(\ref{15ab}).

We now turn our attention to the circumstance that $x \to 0$, which we identify as relating to the form of $\chi(\mu)$  for $\mu \to 0$. For the latter, we see from (\ref{c1}) that
$$
\chi(\mu) \mathop{\sim}\limits_{\mu \to 0^+} \sqrt{\pi \over 2} \bigg ( - \mu + \mu^2 \Big ( 1 - {1 \over \sqrt{2} } \Big ) \bigg ).
$$
Substituting this for $\chi(\mu)$ in the minimisation requirement of (\ref{c2}) gives that the optimal value of $\mu$ is linearly related to $x$. Now evaluating the minimised exponent 
we read off the small
$x$ form of the large $N$ result  (\ref{c2})
$$
p_{N,M}^{\rm r} |_{M=xN^{1/2}} \sim e^{-\sqrt{N} (x -
\sqrt{\pi/2})^2/c}
$$
where $c$ is as in
(\ref{14a}). Now using the relation between $x$ and $\alpha$ as noted below
(\ref{2.33c}) we see that this asymptotic form is precisely the exponential factor of the local CLT
result (\ref{14a}).

\begin{remark} ${}$ \\
  1.~For the intermediate regime in the case of eigenvalue fluctuation in a disk for the Ginibre ensemble, an aspect of the intermediate regime considered in \cite{L+19}  (see also
  \cite{LMS19} and 
  \cite{ABES23}), distinct from the leading order form of the corresponding probability density, was the computation of an analytic expression for the limiting cumulants. \\
  2.~Subsequent to 
  the study of the intermediate regime in the case of eigenvalue fluctuation in a disk for the Ginibre ensemble in \cite{L+19}, it was shown in \cite{FL22} that the asymptotic formula (\ref{c2}) can be exended beyond leading order using the method of mod-phi convergence \cite{FMN16}.
\end{remark}

\subsection{Probability of no real eigenvalues}

From the definitions 
\begin{equation}\label{(23a)}
p_{N,0}^{\rm r} = Z_N(\xi) |_{\xi = 0}.
\end{equation}
Use of Proposition \ref{P2} allows for the leading large $N$ form of $p_{N,0}^{\rm r}$ to be computed.

\begin{prop}
Let $\mu_N$ be as in (\ref{M}). We have
\begin{equation}\label{(22a)}
p_{N,0}^{\rm r} \mathop{\sim}\limits_{N \to \infty} e^{-\sqrt{\pi N/8} \zeta(3/2 ) }
 \mathop{\sim}\limits_{N \to \infty} 
e^{-\mu_N \zeta(3/2)/2} .
\end{equation}
\end{prop}

\begin{proof}
Substituting (\ref{px2}) for the right hand side gives
\begin{equation}\label{(24a)}
\log p_{N,0}^{\rm r} \sim { \sqrt{N \over 2}}
\int_0^\infty \log \Big (
1-  e^{-  t^2} \Big ) \, dt.
\end{equation}
Evaluating the integral, (\ref{(22a)}) follows.
\end{proof}

As a function of $\xi$, we note that (\ref{px2}) is not analytic about $\xi=0$, and so cannot be used to read off
$\log p_{N,M}^{\rm r}$
for $M$ fixed. 
Instead we can return to (\ref{3a}) to reach the conclusion that to leading order the asymptotic form is independent of $M$ and so is given by either of the asymptotic expressions in (\ref{(22a)}).

One notes in (\ref{(22a)}) the appearance of the particular constant $\zeta(3/2)$. In the study of the real eigenvalues of GinOE, this has also appeared in the asymptotic form of the probability of no eigenvalues in an interval of length $s$ in the bulk (replace $N$ by $s$ in the first asymptotic expression of 
(\ref{(22a)}))
\cite{Fo15,FTZ22}.
Most strikingly in the present context is the result from \cite{K+16}
establishing the leading asymptotic form for $p_{N,0}^{\rm r}$ in the case of GinOE is identical to that of the second asymptotic formula in 
(\ref{(22a)}), where the expected number of real eigenvalues is also with respect to GinOE (its leading asymptotics has been noted in the first paragraph of the Introduction), which moreover was extended to the case of the scaled, weakly non-Hermitian version of elliptic GinOE \cite[Th.~1.2 and Remark 1.4]{BMS25}. In fact one reads at the conclusion of \cite[\S 1.2]{BMS25} speculation that this second asymptotic form in (\ref{(22a)}) may indeed be universal.

\subsection*{Acknowledgements}
This work has been supported by the Australian Research Council Discovery Project
DP250102552. The helpful feedback on the paper by Sung-Soo Byun is appreciated.

\providecommand{\bysame}{\leavevmode\hbox to3em{\hrulefill}\thinspace}
\providecommand{\MR}{\relax\ifhmode\unskip\space\fi MR }
\providecommand{\MRhref}[2]{%
  \href{http://www.ams.org/mathscinet-getitem?mr=#1}{#2}
}
\providecommand{\href}[2]{#2}

\end{document}